\documentclass[11pt]{article}

\usepackage{etex}
\usepackage{xspace,enumerate}
\usepackage[dvipsnames]{xcolor}
\usepackage[T1]{fontenc}

\usepackage[american]{babel}
\usepackage{amsfonts}
\usepackage{mathtools}
\usepackage{bm}
\usepackage{amsthm}
\newtheorem{theorem}{Theorem}[section]
\newtheorem*{theorem*}{Theorem}

\newtheorem*{proposition*}{Proposition}
\newtheorem{lemma}[theorem]{Lemma}
\newtheorem*{lemma*}{Lemma}

\newtheorem*{conjecture*}{Conjecture}

\newtheorem*{fact*}{Fact}

\newtheorem*{hypothesis*}{Hypothesis}

\theoremstyle{definition}

\newtheorem*{definition*}{Definition}

\newtheorem*{condition*}{Condition}
\theoremstyle{remark}

\newtheorem*{claim*}{Claim}

\newtheorem*{remark*}{Remark}

\newtheorem*{observation*}{Observation}

\usepackage[
letterpaper,
top=1in,
bottom=1in,
left=1in,
right=1in]{geometry}

\usepackage{paralist}
\usepackage{turnstile}
\usepackage{mdframed}
\usepackage{algorithm}
\usepackage{algpseudocode}
\usepackage{tikz}

\usepackage{microtype}

\usepackage[capitalise,nameinlink]{cleveref}
\crefname{lemma}{Lemma}{Lemmas}
\crefname{fact}{Fact}{Facts}
\crefname{theorem}{Theorem}{Theorems}
\crefname{corollary}{Corollary}{Corollaries}
\crefname{claim}{Claim}{Claims}
\crefname{example}{Example}{Examples}
\crefname{problem}{Problem}{Problems}
\crefname{definition}{Definition}{Definitions}
\crefname{conjecture}{Conjecture}{Conjectures}

\usepackage{boxedminipage}

\newcommand\bdot\bullet

\DeclareMathOperator{\poly}{poly}

\newcommand{\Hoelder}{H\"{o}lder\xspace}

\newcommand{\R}{\mathbb R}

\let\epsilon=\varepsilon
\numberwithin{equation}{section}
\newcommand\MYcurrentlabel{xxx}
\newcommand{\MYstore}[2]{%
  \global\expandafter \def \csname MYMEMORY #1 \endcsname{#2}%
}
\newcommand{\MYload}[1]{%
  \csname MYMEMORY #1 \endcsname%
}
\newcommand{\MYnewlabel}[1]{%
  \renewcommand\MYcurrentlabel{#1}%
  \MYoldlabel{#1}%
}
\newcommand{\MYdummylabel}[1]{}
\newcommand{\torestate}[1]{%
  \let\MYoldlabel\label%
  \let\label\MYnewlabel%
  #1%
  \MYstore{\MYcurrentlabel}{#1}%
  \let\label\MYoldlabel%
}
\newcommand{\restatetheorem}[1]{%
  \let\MYoldlabel\label
  \let\label\MYdummylabel
  \begin{theorem*}[Restatement of \cref{#1}]
    \MYload{#1}
  \end{theorem*}
  \let\label\MYoldlabel
}
\newcommand{\restatelemma}[1]{%
  \let\MYoldlabel\label
  \let\label\MYdummylabel
  \begin{lemma*}[Restatement of \cref{#1}]
    \MYload{#1}
  \end{lemma*}
  \let\label\MYoldlabel
}
\newcommand{\restateprop}[1]{%
  \let\MYoldlabel\label
  \let\label\MYdummylabel
  \begin{proposition*}[Restatement of \cref{#1}]
    \MYload{#1}
  \end{proposition*}
  \let\label\MYoldlabel
}
\newcommand{\restatefact}[1]{%
  \let\MYoldlabel\label
  \let\label\MYdummylabel
  \begin{fact*}[Restatement of \prettyref{#1}]
    \MYload{#1}
  \end{fact*}
  \let\label\MYoldlabel
}
\newcommand{\restate}[1]{%
  \let\MYoldlabel\label
  \let\label\MYdummylabel
  \MYload{#1}
  \let\label\MYoldlabel
}

\allowdisplaybreaks
\newcommand*{\loweredwidetildehelper}[2]{\hbox{\csname dimen@\endcsname\accentfontxheight#1%
  \accentfontxheight#11.25\csname dimen@\endcsname
  $\csname m@th\endcsname#1\widetilde{#2}$%
  \accentfontxheight#1\csname dimen@\endcsname
  }%
}
\newcommand*{\accentfontxheight}[1]{\fontdimen5\ifx#1\displaystyle \textfont \else\ifx#1\textstyle \textfont \else\ifx#1\scriptstyle \scriptfont \else \scriptscriptfont \fi\fi\fi3
}

\title{A Potential Reduction Inspired Algorithm for Exact Max Flow in Almost $\widetilde{O}(m^{4/3})$ Time}

\author{
Tarun Kathuria\thanks{U.C. Berkeley,  {tarunkathuria@berkeley.edu}, supported by NSF Grant CCF 1718695.}
}

\begin{document}

\pagestyle{empty}


\maketitle
\thispagestyle{empty} 


\begin{abstract}

    We present an algorithm for computing $s$-$t$ maximum flows in directed graphs in $\widetilde{O}(m^{4/3+o(1)}U^{1/3})$ time. Our algorithm is inspired by potential reduction interior point methods for linear programming. Instead of using scaled gradient/Newton steps of a potential function, we take the step which maximizes the decrease in the potential value subject to advancing a certain amount on the central path, which can be efficiently computed. This allows us to trace the central path with our progress depending only $\ell_\infty$ norm bounds on the congestion vector (as opposed to the $\ell_4$ norm required by previous works) and runs in $O(\sqrt{m})$ iterations. To improve the number of iterations by establishing tighter bounds on the $\ell_\infty$ norm, we then consider the weighted central path framework of Madry \cite{M13,M16,CMSV17} and Liu-Sidford \cite{LS20}. Instead of changing weights to maximize energy, we consider finding weights which maximize the maximum decrease in potential value. Finally, similar to finding weights which maximize energy as done in \cite{LS20} this problem can be solved by the iterative refinement framework for smoothed $\ell_2$-$\ell_p$ norm flow problems \cite{KPSW19} completing our algorithm. We believe our potential reduction based viewpoint provides a versatile framework which may lead to faster algorithms for max flow. 

\end{abstract}

\clearpage


\setcounter{tocdepth}{1}

  \microtypesetup{protrusion=false}
  \microtypesetup{protrusion=true}

\clearpage

\pagestyle{plain}
\setcounter{page}{1}


\section{Introduction}\label{sec:introduction}

The $s$-$t$ maximum flow problem and its dual, the $s$-$t$ minimum cut on graphs are amongst the most fundamental problems in combinatorial optimization with a wide range of applications.  Furthermore, they serve as a testbed for new algorithmic concepts which have found uses in other areas of theoretical computer science and optimization. This is because the max-flow and min-cut problems demonstrate the prototypical primal-dual relation in linear programs. In the well-known $s$-$t$ maximum flow problem we are given a graph $G=(V,E)$ with $m$ edges and $n$ vertices with edge capacities $u_e \leq U$, and aim to route as much flow as possible from $s$ to $t$ while restricting the magnitude of the flow on each edge to its capacity.\\

Several decades of work in combinatorial algorithms for this problem led to a large set of results culminating in the work of Goldberg-Rao \cite{GR98} which gives a running time bound  of $O(m \min\{m^{1/2},n^{2/3}\}\log(\frac{n^2}{m})\log U)$. This bound remained unimproved for many years. In a breakthrough paper, Christiano et al \cite{CKMST11} show how to compute approximate maximum flows in $\widetilde{O}(mn^{1/3}\log(U)\mathsf{poly}(1/\varepsilon))$. Their new approach uses electrical flow computations which are Laplacian linear system solves which can be solved in nearly-linear time \cite{ST14} to take steps to minimize a softmax approximation of the congestion of edges via a second order approximation. A straightforward analysis leads to a $O(\sqrt{m})$ iteration algorithm. However, they present an insight by trading off against another potential function and show that $O(m^{1/3})$ iterations suffice. This work led to an extensive line of work exploiting Laplacian system solving and continuous optimization techniques for faster max flow algorithms. Lee et al. \cite{LRS13} also present another $O(n^{1/3}\poly(1/\varepsilon))$ iteration algorithm for unit-capacity graphs also using electrical flow primitives. Finally Kelner et al. \cite{KLOS14} and Sherman \cite{Sherman13,Sherman17a} present algorithms achieving $O(m^{o(1)}\poly(1/\varepsilon))$ iteration algorithm for max-flow and its variants, which are based on congestion approximators and oblivious routing schemes as opposed to electrical flow computations. This has now been improved to near linear time \cite{Peng16,Sherman17b}. Crucially this line of work can only guarantee weak approximations to max flow due to the $\poly(1/\varepsilon)$ in the iteration complexity.\\

In order to get highly accurate solutions which depend only polylogarithmically on $1/\varepsilon$, work has relied on second-order optimization techniques which use first and second-order information (the Hessian of the optimization function). To solve the max flow problem to high accuracy, several works have used interior point methods (IPMs) for linear programming \cite{NN94,Ren01}. These algorithms approximate non-negativity/$\ell_\infty$ constraints by approximating them by a \textit{self-concordant} barrier, an approximation to an indicator function of the set which satisfies local smoothness and strong convexity properties and hence can be optimized using Newton's method. In particular, Daitch and Spielman \cite{DS08} show how to combine standard path-following IPMs and Laplacian linear system solves to obtain $\widetilde{O}(m\sqrt{m}\log (U/\varepsilon))$ iterations, matching Goldberg and Rao up to logarithmic factors. The $O(\sqrt{m})$ iterations is a crucial bottleneck here due to the $\ell_\infty$ norm being approximated by $\ell_2$ norm to a factor of $\sqrt{m}$. Then Lee and Sidford \cite{LS14} devised a faster IPM using weighted logarithmic barriers to achieve a $\widetilde{O}(m\sqrt{n}\log(U/\varepsilon)$ time algorithm. Madry \cite{M13,M16} opened up the weighted barriers based IPM algorithms for max flow to show that instead of $\ell_2$ norm governing the progress of each iteration, one can actually make the progress only maintaining bounds on the $\ell_4$ norm. Combining this with insights from \cite{CKMST11}, by using another potential function, which again depends on the energy of the next flow step and carefully tuning the weights in the barriers, he achieved an  $\widetilde{O}(m^{3/7})$ iteration algorithm which leads to a $\widetilde{O}(m^{11/7}U^{1/7}\log(m/\varepsilon))$ time. Note that the algorithm depends polynomially on the maximum capacity edge $U$ and hence is mainly an improvement for mildly large edge capacities. This work can also be used to solve min cost flow problems in the same running time \cite{CMSV17}.\\

Another line of work beyond IPMs is to solve $p$-norm regression problems on graphs. Such problems interpolate between electrical flow problems $p=2$, maximum flow problems $p=\infty$ and transshipment problems $p=1$. While these problems can also be solved in $O(\sqrt{m})$ iterations to high accuracy using IPMs\cite{NN94}, it was unclear if this iteration complexity could be improved depending on the value of $p$. Bubeck et al. \cite{BCLL18} showed that for any self-concordant barrier for the $\ell_p$ ball, the iteration complexity has to be at least $O(\sqrt{m})$ thus making progress using IPMs unlikely. They however showed another \textit{homotopy-based} method, of which IPMs are also a part of, can be used to solve the problem in $\widetilde{O}_p(m^{\frac{1}{2}-\frac{1}{p}}\log(1/\varepsilon))$ iterations, where $O_p$ hides dependencies on $p$ in the runtime. This leads to improvements on the runtime for constant values of $p$. Next, Adil et al. \cite{AKPS19}, inspired by the work of \cite{BCLL18} showed that one can measure the change in $p$-norm using a second order term based on a different function which allows them to obtain approximations to the $p$-norm function in different norms with strong condition number. These results can be viewed in the framework of relative convexity \cite{LFN18}. Thus, they can focus on just solving the optimization problem arising from the residual. Using insights from $\cite{CKMST11}$, they arrive at a $\widetilde{O}_p(m^{4/3}\log(1/\varepsilon)$-time algorithm. Then follow-up work by Kyng et al. \cite{KPSW19} opened up the tools used by Spielman and Teng \cite{ST14} for $\ell_2$-norm flow problems to show that one can construct strong preconditioners for the residual problems for mixed $\ell_2$-$\ell_p$-norm flow problems, a generalization of $\ell_p$-norm flow and obtain an $\widetilde{O}_p(m^{1+o(1)}\log(1/\varepsilon)$ algorithm. These results however do not lead to faster max flow algorithms however due to their large dependence on $p$.\\

However, Liu and Sidford \cite{LS20} improving on Madry \cite{M16} showed that instead of carefully tuning the weights based on the electrical energy, one can consider the separate problem of finding a new set of weights under a certain budget constraint to maximize the energy. They showed that a version of this problem reduce to solving $\ell_2$-$\ell_p$ norm flow problems and hence can be solved in almost-linear time using the work of \cite{KPSW19,AS20}. This leads to a $O(m^{11/8+o(1)}U^{1/4})$-time algorithm for max flow. However, this result still relies on the amount of progress one can take in each iteration being limited to the bounds one can ensure on the $\ell_4$ norm of the congestion vector, as opposed to the ideal $\ell_\infty$ norm.  We remark here that there are IPMs for linear programming which only measure centrality in $\ell_\infty$ norm as opposed to the $\ell_2$ or $\ell_4$ norm. In particular \cite{CLS19,LSZ19,BLSS20} show how to take a step with respect to a softmax function of the duality gap and trace the central path only maintaining $\ell_\infty$ norm bounds. \cite{Tuncel95,Tuncel94} also designed potential reduction based IPMs which trace the central path only maintaining centrality in $\ell_\infty$.

\subsection{Our Contribution}
In this paper, we devise a faster interior point method for $s$-$t$ maximum flow in directed graphs. Precisely, our algorithm runs in time $\widetilde{O}(m^{4/3+o(1)}U)$. During the process of writing this paper, we were informed by Yang Liu and Aaron Sidford \cite{LS20b} that they have also obtained an algorithm achieving the same runtime. They also end up solving the same subproblems that we will end up solving, although they arrive at it from the perspective of considering the Bregman divergence of the barrier as opposed to considering the potential funcion that is the inspiration for our work.  Our algorithm builds on top of both Madry \cite{M16} and Liu-Sidford \cite{LS20} and is arguably simpler than both in some regards. \\

In particular, our algorithm is based on potential reduction algorithms which are a kind of interior point methods for linear programs. These algorithms are based on a potential function which measures both the duality gap as well as accounts for closeness to the boundary via a barrier function. The algorithms differ from path-following IPMs in that they have the potential to not strictly follow the path closely but only trace it loosely, which is also experimentally observed. Usually, the step taken is a scaled gradient step/Newton step on the potential function. Provided that we can guarantee sufficient decrease of the potential function and relate the potential function to closeness to optimality, we can show convergence. We refer to \cite{Ans96,Todd96,NN94} for excellent introductions to potential reduction IPMs. \\

We will however use a different step; instead of a Newton step, we consider taking the step, subject to augmenting a certain amount of flow in each iteration, which maximizes the decrease in the potential function after taking the step. We then show that this optimization problem can be efficiently solved in $\widetilde{O}(m)$ time using electrical flow computations. While we can show that the potential function decreases by a large amount which guarantees that we can solve the max flow problem in $O(\sqrt{m})$ iterations, we forego writing it in this manner as we are unable to argue such a statement when the weights and hence the potential function is also changed. Instead, we stick to keeping track of the centrality of our flow vector while making sufficient progress. Crucially however, the amount of progress made by our algorithm only depends on bounds on the $\ell_\infty$ of the congestion vector of the update step rather than the traditional $\ell_2$ or $\ell_4$ norm bounds in \cite{M16,LS20}.  In order to improve the iteration complexity by obtaining stronger bounds on the $\ell_\infty$ norm of the congestion vector, we show that like in Liu-Sidford \cite{LS20}, we can change weights on the barrier term for each edge. Instead of using energy as a potential function to be maximized, inspired by oracles designed for multiplicative weights algorithms, we use the change in the potential function itself as the quantity to be maximized subject to a $\ell_1$ budget constraint on the change in weights. While we are unaware of how to maximize the $\ell_1$ constrained problem, we relax it to an $\ell_q$ constrained problem, which we solve using a mixed $\ell_2$-$\ell_p$ norm flow problem using the work of \cite{KPSW19,AS20}. Combining this with an application of \Hoelder's inequality gives us sufficiently good control on the $\ell_1$-norm of the weight change while ensuring that our step has significantly better $\ell_\infty$ norm bounds on the congestion vector. We believe our potential reduction framework as well as the concept of changing weights based on the update step might be useful in designing faster algorithms for max flow beyond our $m^{4/3}$ running time.

\section{Preliminaries}\label{sec:prelims}
Throughout this paper, we will view graphs as having both forward and backward capacities. Specifically, we will denote by $G=(V,E,\bf{u})$, a directed graph with vertex set $V$ of size $n$, an edge set $E$ of size $m$, and two non-negative capacities $u_e^-$ and $u_e^+$ for each edge $e\in E$. For the purpose of this paper, all edge capacities are bounded by $U=1$. Each edge $e=(u,v)$ has a head vertex $u$ and a tail vertex $v$. For a vector $v \in \R^m$, we define $\|v\|_p=(\sum\limits_{i=1}^{m}|v_i|^p)^{1/p}$ and $\|v\|_\infty = \max\limits_{i=1}^{m}|v_i|$ and refer to $\text{Diag}(v)\in\R^{m \times m}$ as the diagonal matrix with the $i^{th}$ diagonal entry equal to $v_i$.

\textbf{Maximum Flow Problem} Given a graph $G$, we call any assignment of real values to the edges of $E$, i.e., $f\in\R^m$, a flow. For a flow vector $f$, we view $f_e$ as the amount of the flow on edge $e$ and if this value is negative, we interpret it as having a flow of $|f_e|$ flowing in the direction opposite to the edge's orientation. We say that a flow $f$ is an $\sigma$-flow, for some demands $\sigma\in\R^n$ iff it satisfies \textit{flow conservation constraints} with respect to those demands. That is, we have
\begin{align*}
    \sum\limits_{e \in E^+(v)}f_e-\sum\limits_{e \in E^-(v)}f_e = \sigma_v \ \text{for every vertex } v \in V
\end{align*}
 where $E^+(v)$ and $E^-(v)$ is the set of edges of $G$ that are entering and leaving vertex $v$ respectively. We will require $\sum\limits_{v \in V} \sigma_v=0$.

 Furthermore, we say that a $\sigma$-flow $f$ is feasible in $G$ iff $f$ satisfies the capacity constraints
 \begin{align*}
     -u_e^-\leq f_e \leq u_e^+ \ \text{for each edge } e \in E
 \end{align*}
 One type of flows that will be of interest to us are $s-t$ flows, where $s$ (the \textit{source}) and $t$(the \textit{sink}) are two distinguishing vertices of G. Formally, an $s-t$ flow is a $\sigma$-flow whose demand vector $\sigma=F\chi_{s,t}$, where $F$ is the value of the flow and $\chi_{s,t}$ is a vector with $-1$ and $+1$ at the coordinates corresponding to $s$ and $t$ respectively and zero elsewhere.

 Now, the maximum flow problem corresponds to the problem in which we are given a directed graph $G=(V,E,u)$ with integer capacities as well as a source vertex $s$ and a sink vertex $t$ and want to find a feasible $s$-$t$ flow of maximum value. We will denote this maximum value $F^*$

 \textbf{Residual Graphs} A fundamental object in many maximum flow algorithms is the notion of a residual graph. Given a graph $G$ and a feasible flow $\sigma$-flow $f$ in that graph, we define the \textit{residual graph} $G_f$ as a graph $G=(V,E,\hat{u}(f))$ over the same vertex and edge set as $G$ and such that, for each edge $e=(u,v)$, it's forward and backward residual capacities are defined as
 \begin{align*}
     \hat{u}^+_e(f)=u_e^+ - f_e \text{ and } \hat{u}^-_e(f)=u_e^- + f_e
 \end{align*}
 We will also denote $\hat{u}_e(f) = \min\{\hat{u}^+_e(f),\hat{u}^-_e(f)\}$. When the value of $f$ is clear from context, we will omit writing it explicitly. Observe that the feasibility of $f$ implies that all residual capacities are always non-negative.

 \textbf{Electrical Flows and Laplacian Systems} Let $G$ be a graph and let $r\in\R^m_{++}$ be a vector of edge resistances, where the resistance of edge $e$ is denoted by $r_e$. For a flow $f \in \R^E$ on $G$, we define the energy of $f$ to be $\mathcal{E}_r(f) = f^\top R f = \sum\limits_{e \in E} r_e f_e^2$ where $R = \text{Diag}(f)$. For a demand $\chi$, we define the electrical $\chi$-flow $f_r$ to be the $\chi$-flow which minimizes energy $f_r = \arg\min\limits_{B^\top f=\chi} \mathcal{E}_r(f)$, where $B\in \R^{m \times n}$ is the edge-vertex incidence matrix. This flow is unique as the energy is a strcitly convex function.
 The Laplacian of a graph $G$ with resistances $r$ is defined as $L=B^\top R^{-1} B$. The electrical $\chi$ flow is given by the formula $f_r = R^{-1}BL^{\dagger}\chi$. We also define electrical potentials as $\phi=L^{\dagger}\chi$ There is a long line of work starting from Spielman and Teng which shows how to solve $L\phi = \chi$ in nearly linear time \cite{ST14,KMP14,KOSZ13,PS14,CKMPPRX14,KS16,KLPSS16}.

 \textbf{p-Norm Flows} As mentioned above, a line of work \cite{BCLL18,AKPS19,KPSW19} shows how to solve more general $p$-norm flow problems. Precisely, given a "gradient" vector $g \in \R^E$, resistances $r \in \R_{+}^E$ and a demand vector $\chi$, the problem under consideration is
 \begin{align*}
     OPT=\min\limits_{B^\top f = \chi} \sum\limits_{e \in E} g_e f_e + r_e f_e^2 + |f_e|^p
 \end{align*}
 \cite{KPSW19} call such a problem as a mixed $\ell_2$-$\ell_p$-norm flow problem and denote the expression inside the min as $val(f)$. The main result of the paper is
 \begin{theorem}[Theorem 1.1 in \cite{KPSW19}]\label{thm:kpswthm} For any even $p\in [\omega(1), o(\log^{2/3-o(1)} n)]$ and an initial solution $f^{(0)}$ such that all parameters are bounded by $2^{\poly(\log(n))}$, we can compute a flow $\widetilde{f}$ satisfying the demands $\chi$ such that
 \begin{align*}
     val(\widetilde{f}) - OPT \leq \frac{1}{2^{O(\poly (\log m))}}(val(f^{(0)}) - OPT) + \frac{1}{2^{O(\poly (\log m))}}
 \end{align*}
 in $2^{O(p^{3/2})}m^{1+O(1/\sqrt{p})}$ time.
 \end{theorem}
 We remark that strictly speaking the theorem in \cite{KPSW19} states the error to be polynomial but \cite{LS20} observe that their proof actually implies quasi-polynomial error as stated above.
 While our subproblems that we need to solve to change weights cannot be exactly put into this form, we show that mild modifications to their techniques can be done to then use their algorithm as a black-box. Hence, we elaborate on their approach below.

 One main thing to establish in their paper is how the $p$-norm changes when we move from $f$ to $f+\delta$.
 \begin{lemma}[Lemma in \cite{KPSW19}]\label{lem:kpsw} We have for any $f \in \R^E$ and $\delta\in\R^E$ that
 \begin{align*}
 f_i^p + p f_i^{p-1}\delta_i+2^{-O(p)} h_p(f_i^{p-2},\delta_i) \leq (f_i + \delta_i)^p \leq f_i^p + p f_i^{p-1}\delta_i+2^{O(p)} h_p(f_i^{p-2},\delta_i)
 \end{align*}
 where $h_p(x,\delta) = x\delta^2 + \delta^p$
 \end{lemma}
 Hence, given an initial solution, it suffices to solve the residual problem of the form
 \begin{align*}
     \min\limits_{B^\top f = 0} g(f)^\top \delta + \sum\limits_{e\in E} h_p(f^{p-2}_i,\delta_i)
 \end{align*}
 where $g(f)_i = pf_i^{p-1}$. Next, they notice that bounding the condition number with respect to the function $h_p(\cdot,\cdot)$ actually suffices to get linear convergence and hence tolerate quasi-polynomially low errors. The rest of the paper goes into designing good preconditioners which allow them to solve the above subproblem quickly.

 We will also need some basics about min-max saddle point problems \cite{BNO03}. Given a function $f(x,y)$ such that $\mathsf{dom}(f,x) = \mathcal{X}$ and $\mathsf{dom}(f,y) = \mathcal{Y}$. The problem we will be interested in is of the form
 \begin{align*}
     \min\limits_{x \in \mathcal{X}}\max\limits_{y \in \mathcal{Y}}f(x,y)
 \end{align*}
 Define the functions $f_y(y) = \min\limits_{x \in \mathcal{X}} f(x,y)$ and $f_x(x) = \max\limits_{y \in \mathcal{Y}} f(x,y)$ for every fixed $y\in\mathcal{Y}$. We have the following theorem from Section 2.6 in \cite{BNO03}
 \begin{theorem}\label{thm:minmaxthm} If $f(x,y)$ is convex in $x$ and concave in $y$ and let $\mathcal{X,Y}$ be convex and closed. Then $f_x$ is a convex function and $f_y$ is a concave function. \end{theorem}

\section{Warm up : $\sqrt{m}$ Iteration Algorithm}\label{sec:warmup}
In this section, we first set up our IPM framework and show how to recover the $\sqrt{m}$ iterations bound for max flow. In the next section, we will then change the weights to obtain our improved runtime. Our framework is largely inspired by \cite{M16} and \cite{LS20} and indeed a lot of the arguments can be reused with some modifications.
\subsection{IPM Setup}For every edge $e=(u,v)$, we consider assigning two non-negative weights for the forward and backward edges $w_e^+$ and $w_e^-$. Based on the weights and the edge capacities, for any feasible flow, we define a barrier functional
\begin{align*}
    \phi_w(f) = -\sum\limits_{e \in E} w_e^+ \log(u_e^+ - f_e) + w_e^- \log(u_e^- + f_e)
\end{align*}
IPMs iterate towards the optimal solution by trading off the amount of progress of the current iterate, i.e., $B^\top f = F\chi$ and the proximity of the point to the constraints measured through the barrier $\phi_w(f)$, known as centrality. Previous IPMs taking a Newton step with respect to the barrier with a size which ensures that we increase the value of the flow $F$ by a certain amount. Due to the fact that a Newton step is the minimization of a second order optimization problem, it can be shown that the step can be computed via electrical flow computations. Typically, taking a Newton step can be decomposed into progress and centering steps where one first takes a progress step which increases the flow value which causes us to lose centrality by some amount. Then one takes a centering step which improves the centrality without increasing the flow value. Depending on the amount of progress we can make in each iteration such that we can still recenter determines the number of iterations our algorithm will take. \cite{M16,LS20} follow this prototype and loosely speaking  the amount of flow value we can increase in each iteration for the progress step depends on the $\ell_\infty$ norm of the congestion vector, which measures how much flow we can add before we saturate an edge. However, the bottleneck ends up being the centering step which requires that the flow value can only be increased by an amount depending on the $\ell_4$ norm of the congestion vector which is a stronger condition than $\ell_\infty$ norm.

\cite{M13,M16} notes that when the $\ell_\infty$ and $\ell_4$ norms of the congestion vector are large then increasing the resistances of the congested edges increases the energy of the \textit{resulting} electrical flow. So he repeatedly increases the weights of the congested edges (called boosting) until the congested vector has sufficiently small norm. By using electrical energy of the resulting step as a global potential function and analyzing how it evolves over the progress, centering and boosting steps, they can control the amount of weight change and number of boosting steps necessary to reduce the norm of the congestion vector. Carefully trading these quantities yields their runtime of $\widetilde{O}(m^{11/7})$. To improve on this, Liu and Sidford \cite{LS20} consider the problem of finding a set of weight increases which maximize the energy of the resulting flow. As we need to ensure that the weights don't increase by too much, they place a budget constraint on the weight vector. By showing that a small amount of weight change suffices to obtain good bounds on the congestion vector. Fortunately, this optimization problem ends up being efficiently solvable in almost linear time by using the mixed $\ell_2$-$\ell_p$ norm flow problem of \cite{KPSW19}. However, this step still essentially requires $\ell_4$-norm bounds to ensure centering is possible.

In this paper, we will consider taking steps with respect to a potential function. The potential function $\Phi_w$  comes from potential reduction IPM schemes and trades off the duality gap with the barrier.
\begin{align*}
    \Phi_w(f,s)=m\log\left(1+\frac{f^{\top}s}{m}\right)+\phi_{w}(f)
\end{align*}
For self-concordant barriers like weighted log barriers are, the negative gradient $-\nabla\phi_w(f)$ is feasible for the dual \cite{Ren01} and so for any $f'$ feasible for the primal, we have $f'^\top (-\nabla\phi_w(f))\geq 0$. We will consider dual "potential" variables $y\in \R^V$. Now, like in \cite{M16,LS20}, we consider a centrality condition
\begin{equation}
    y_v - y_u = \frac{w_e^+}{u_e^+-f_e} - \frac{w_e^-}{u_e^-+f_e} \text{ for all } e=(u,v)
\end{equation}
If $(f,y,w)$ satsify the above condition, we call it \textit{well-coupled}.
Also, given a tuple $(f,y,w)$ and a candidate step $\hat{f}$, define the forward and backward congestion vectors $\rho^+,\rho^-\in \R^E$ as
\begin{align}
    \rho^+_e = \frac{|\hat{f}_e|}{u_e^+-f_e} \text{ and } \rho^-_e = \frac{|\hat{f}_e|}{u_e^-+f_e} \text{ for all } e \in E
\end{align}

We can now assume via binary search that we know the optimal flow value $F^*$ \cite{M16}.
\cite{M16,LS20} consider preconditioning the graph which allows them to ensure that for a well-coupled point we can ensure sufficient progress. The preconditioning strategy to ensure this is to add $m$ extra (undirected) edges between $s$ and $t$ of capacity $2U$ each. So the max flow value increases at most by $2mU$. The following lemma can be seen from the proof of Lemma 4.5 in \cite{LS20}
\begin{theorem}\label{lem:precond} Let $(f,y,w)$ be a well-coupled point for flow value $F$ in a preconditioned graph $G$. Then we have for every preconditioned edge $e$ that $\hat{u}_e(f) = \min\{u_e^+-f_e,u_e^-+f_e\}\geq \frac{F^*-F}{7\|w\|_1}$. In particular, if $\|w\|_1\leq 3m$, then we have $\hat{u}_e(f) \geq \frac{F^*-F}{21m}$. If we also have $F^* - F \geq m^{1/2-\eta}$, then $\hat{u}_e(f) \geq m^{-(1/2+\eta)}/21$
\end{theorem}
Now that our setup is complete, we can focus on the step that we will be taking. In this section, we will keep the weights all fixed to 1, i.e., $w_e^+ = w_e^- = 1$ for all $e \in E$. Hence $\|w\|_1 = 2m$. Consider the change in the potential function when we move from $f$ to $f+\hat{f}$ while keeping the dual variable $-\nabla\phi_w(f) = By$ fixed. This change is
\begin{align*}
   m\log\left(1-\frac{(f+\hat{f})^{\top}\nabla\phi_{w}(f)}{m}\right)-m\log\left(1-\frac{f^{\top}\nabla\phi_{w}(f)}{m}\right)+\phi_{w}(f+\hat{f})-\phi_{w}(f)
\end{align*}
We are interested in minimizing this quantity which corresponds to maximizing the decrease in the potential function value while guaranteeing that we send say $\delta$ more units of flow $\hat{f}$. Hence the problem is
\begin{align*}
     \arg\min\limits_{B^\top \hat{f}=\delta \chi} m\log\left(1-\frac{(f+\hat{f})^{\top}\nabla\phi_{w}(f)}{m}\right) + \phi_w(f+\hat{f})
\end{align*}
Unfortunately, this problem is not convex as the duality gap term is concave in $\hat{f}$. However, we instead can minimize an upper bound to this term which is convex:
\begin{align*}
    &\arg\min\limits_{B^\top \hat{f}=\delta \chi}\phi_w(f+\hat{f}) - (f+\hat{f})^\top\nabla\phi_w(f)
    \\
    &= \arg\min\limits_{B^\top \hat{f}=\delta \chi}-\sum\limits_{e \in E} w_e^+\log\left(1-\frac{\hat{f}_e}{u_e^+-f_e}\right) + w_e^-\log\left(1+\frac{\hat{f}_e}{u_e^-+f_e}\right) - \hat{f}_e\left(\frac{w_e^+}{u_e^+-f_e} - \frac{w_e^-}{u_e^-+f_e}\right)
\end{align*}
as $\log(1+x) \leq x$ for non-negative $x$ which holds from duality as mentioned above. We will refer to the value of the problem in the last line as the \textit{potential decrement} and will henceforth denote the function inside the minimization as $\Delta\Phi_w(f,\hat{f})$. It is instructive to first see how the coupling condition changes if we were to take the optimal step of the above problem, while remaining feasible. To calculate this, from the optimality conditions of the above program, we can say that there exists a $\hat{y}$ such that for all $e=(u,v)$
\begin{align*}
    \hat{y}_v-\hat{y}_u &= \left(\frac{w_e^+}{u_e^+-f_e-\hat{f}_e} - \frac{w_e^-}{u_e^-+f_e+\hat{f}_e}\right) -  \left(\frac{w_e^+}{u_e^+-f_e} - \frac{w_e^-}{u_e^-+f_e}\right)\\
    &=\left(\frac{w_e^+}{u_e^+-f_e-\hat{f}_e} - \frac{w_e^-}{u_e^-+f_e+\hat{f}_e}\right) -  (y_v-y_u)
\end{align*}
Hence, if we update $y$ to $y+\hat{y}$ and $f$ to $f+\hat{f}$, we get a flow of value $F+\delta$ such that the coupling condition with respect to the new $y$ and $f$ still hold.

Hence, we can now focus on actually computing the step and showing what $\delta$ we can take to ensure that we still satisfy feasibility, i.e., bounds on the $\ell_\infty$ norm of the congestion vector. The function we are trying to minimize comprises of a self-concordant barrier term and a linear term. Unfortunately, we cannot control the condition number of such a function to optimize it in efficiently over the entire space as this is arguably as hard as the original problem itself. However, due to self-concordance, the function behaves smoothly enough (good condition number) in a box around the origin but that seemingly doesn't help us solve the problem over the entire space. Fortunately, a fix for this was already found in \cite{BCLL18}. In particular they (smoothly) extend the function quadratically outside a box to ensure that the (global) smoothness and strong convexity properties inside the box carries over to that outside the box as well while still arguing that the minimizer is the same provided the minimizer of the original problem was inside the box. Specifically, the following lemma can be inferred from Section 2.2 of \cite{BCLL18}.
\begin{lemma}
Given a function $f(x)$ which is $L$-smooth and $\mu$-strongly convex inside an interval $[-\ell,\ell]$. Then, we define the quadratic extension of $f$, defined as \[
  f_\ell(x) = \left.
  \begin{cases}
    f(x), & \text{for } -\ell \leq x \leq \ell \\
    f(-\ell)+f'(-\ell)(x+\ell)+\frac{1}{2}f''(-\ell)(x+\ell)^2, & \text{for } x < -\ell \\
    f(\ell)+f'(\ell)(x-\ell)+\frac{1}{2}f''(\ell)(x-\ell)^2, & \text{for } x> \ell
  \end{cases}
  \right\}
\]
The function $f_\ell$ is $C^2$, $L$-smooth and $\mu$-strongly convex. Furthermore, for any convex function $\psi(x)$ provided $x^* = \arg\min\limits_{x \in \mathcal{X}}\psi(x) + \sum\limits_{i=1}^{n}f(x_i)$ lies inside $\prod\limits_{i=1}^{n}[-\ell_i,\ell_i]$, then $\arg\min\limits_{x \in \mathcal{X}} \psi(x) + \sum\limits_{i=1}^{n} f_{\ell_i}(x_i) = x^*$
\end{lemma}
Hence, it suffices to consider a $\delta$ small enough such that the minimizer is the same as for the original problem and we can focus on minimizing this quadratic extension of the function. For minimization, we can use Accelerated Gradient Descent or Newton's method.
\begin{theorem}[\cite{Nes04}]\label{thm:agd} Given a convex function $f$ which satisfies $D \preceq \nabla^2 f(x) \preceq \kappa D \forall x \in R^n$ with some given fixed diagonal matrix $D$ and some fixed $\kappa$. Given an initial point $x_0$ and
an error parameter $0 < \varepsilon < 1/2$, the accelerated gradient descent (AGD) outputs x such that
\begin{align*}
    f(x) - \min\limits_x f(x) \leq \varepsilon(f(x_0) - \min\limits_x f(x))
    \end{align*}
in $O(\sqrt{\kappa} \log(\kappa/\varepsilon))$ iterations. Each iteration involves computing $\nabla f$ at some point x and projecting the function onto the subspace defined by the constraints and some linear-time calculations.
\end{theorem}
Notice that the Hessian of the function in the potential decrement problem is a diagonal matrix with the $e^{th}$ entry being $$\frac{w_e^+}{(u_e^+-f_e-\hat{f}_e)^2}+\frac{w_e^-}{(u_e^-+f_e+\hat{f}_e)^2}$$ So provided $\rho^+_e,\rho^-_e$ are less than some small constant, the condition number $\kappa$ of the Hessian is constant with respect to the diagonal matrix which is $\nabla^2\phi_w(f)$ and hence we can use Theorem \ref{thm:agd} to solve it in $\widetilde{O}(1)$ to quasi-polynomially good error. Furthermore notice that the algorithm is just computing a gradient and then doing projection and so can be computing using a Laplacian linear system solve and hence runs in nearly linear time. Furthermore, quasi-polynomially small error will suffice for our purposes \cite{M13,M16,LS20}.

Now, we just need to ensure that we can control the $\ell_\infty$-norm of the congestion vector, as that controls how much flow we can still send without violating constraints. Note further, that we need to set $\ell$ while solving the quadratic extension of the potential decrement problem so that it's greater than the $\ell_\infty$ norm that we can guarantee. We will want both of these to be some constants.

As mentioned above, the point of preconditoning the graph is to ensure that the preconditioned edges themselves can facilitiate sufficient progress. To bound the congestion, we show an analog of Lemma 3.9 in \cite{M16}.
\begin{lemma}\label{lem:constcongrootm}
Let $(f,y,w)$ be a well-coupled solution with value $F$ and let $\delta=\frac{F^*-F}{1000\sqrt{m}}$. Let $\hat{f}$ be the solution to the potential decrement problem. Then we have, $\rho_e^+,\rho_e^-\leq 0.1$ for all edges $e$.
\end{lemma}
\begin{proof}
Consider a flow $f'$ which sends $\frac{2\delta}{m}$ units of flow on each of the $m/2$ preconditioned edges. Certainly the potential decrement flow $\hat{f}$ will have smaller potential decrement than that of $f'$ which is
\begin{align*}
    \Delta\Phi_w(f,f')&= -\sum\limits_{e \in E} w_e^+\log\left(1-\frac{f'_e}{u^+_e-f_e}\right) + w_e^- \log\left(1+\frac{f'_e}{u^+_e-f_e}\right) - f'_e\left(\frac{w_e^+}{u_e^+-f_e} - \frac{w_e^-}{u_e^-+f_e}\right)\\
    &\leq \sum\limits_{e \in E} w_e^+\left(\frac{f'_e}{\hat{u}_e^+(f)}\right)^2 + w_e^-\left(\frac{f'_e}{\hat{u}_e^-(f)}\right)^2\\
    &\leq \|w\|_1 \left(\frac{42\delta}{F^*-F}\right)^2\\
    &< \frac{0.002\|w\|_1}{m}\leq 0.004
\end{align*}
where the second inequality follows from $-\log(1-x) \leq x+x^2$ and $-\log(1+x)\leq -x+x^2$ for non-negative $x$ and the third inequality follows from plugging in the value of the flow on the preconditioned edges and using Lemma \ref{lem:precond}. Finally we use $\|w\|_1=2m$.
Now it suffices to prove a lower bound on the potential decrement in terms of the congestion vector. For this, we start by considering the inner product of $\hat{f}$ with the gradient of the $\Delta\Phi_w(f,\hat{f})$
\begin{align*}
    \sum\limits_{e \in E} \hat{f}_e \left(\frac{w_e^+}{u_e^+-f_e-\hat{f}_e} - \frac{w_e^-}{u_e^-+f_e+\hat{f}_e} - \frac{w_e^+}{u_e^+-f_e} + \frac{w_e^-}{u_e^-+f_e}\right)&= \sum\limits_{e \in E}  \left(\frac{w_e^+\hat{f}_e^2}{(\hat{u}_e^+-\hat{f}_e)\hat{u}^+_e} + \frac{w_e^-\hat{f}_e^2}{(\hat{u}_e^-+\hat{f}_e)\hat{u}_e^-}\right)\\
   & \leq \sum\limits_{e \in E}  1.1\left(\frac{w_e^+\hat{f}_e^2}{(\hat{u}^+_e)^2} + \frac{w_e^-\hat{f}_e^2}{(\hat{u}_e^-)^2}\right)\end{align*}
   \begin{align*}
   &\leq 1.1\sum\limits_{e \in E}  \left(\frac{w_e^+\hat{f}_e^2}{(\hat{u}^+_e)^2} + \frac{w_e^-\hat{f}_e^2}{(\hat{u}_e^-)^2}\right)\\
&\leq 2.2\sum\limits_{e \in E} -w_e^+\log\left(1-\frac{f'_e}{\hat{u}^+_e}\right) - w_e^- \log\left(1+\frac{f'_e}{\hat{u}^+_e}\right) - f'_e\left(\frac{w_e^+}{\hat{u}_e^+} - \frac{w_e^-}{\hat{u}_e^-}\right)\\
&= 2.2\Delta\Phi_w(f,\hat{f})\\
&\leq 0.0088
    \end{align*}
    where the second-to-last inequality follows from $x+x^2/2\leq -\log(1-x)$ and $-x+x^2/2\leq -\log(1+x)$. Strictly speaking, the first inequality only holds for $\hat{f}_e \leq \hat{u}_e(f)/10$. However, instead of considering the inner product of $\hat{f}$ with the gradient of $\Delta\Phi_w(f,\hat{f})$, we will instead consider it's quadratic extension with $\ell_e=\hat{u}_e(f)/10$ for each edge $e$. It is easy to see that if $\hat{f}$ is outside the box, then also the desired inequality still holds (by computing the value the quadratic extension takes on $f'$ in the cases outside the box).
    To finish the proof,
    \begin{align*}
        \sum\limits_{e \in E} \hat{f}_e \left(\frac{w_e^+}{u_e^+-f_e-\hat{f}_e} - \frac{w_e^-}{u_e^-+f_e+\hat{f}_e} - \frac{w_e^+}{u_e^+-f_e} + \frac{w_e^-}{u_e^-+f_e}\right)&= \sum\limits_{e \in E}  \left(\frac{w_e^+\hat{f}_e^2}{(\hat{u}_e^+-\hat{f}_e)\hat{u}^+_e} + \frac{w_e^-\hat{f}_e^2}{(\hat{u}_e^-+\hat{f}_e)\hat{u}_e^-}\right)\\
        &\geq 9/10 \sum\limits_{e \in E}  \left(\frac{w_e^+\hat{f}_e^2}{(\hat{u}^+_e)^2} + \frac{w_e^-\hat{f}_e^2}{(\hat{u}_e^-)^2}\right)\\
        &\geq 0.9\|\rho\|_\infty^2
    \end{align*}
    Hence, combining the above, we get that $\|\rho\|_\infty\leq 0.1$
\end{proof}

Notice that since $\|\rho\|_\infty<0.1$, the minimizer of the quadratic smoothened function is the same as the function without smoothing and hence the new step is well-coupled as per the argument above. Hence, in every iteration, we decrease the amount of flow that we could send multiplicatively by a factor of $1-1/\sqrt{m}$ and hence in $\sqrt{m}$ iterations we will get to a sufficiently small amount of remaining flow that we can round using one iteration of augmenting paths. This completes our $\sqrt{m}$ iteration algorithm.
\section{Improved $m^{4/3+o(1)}U^{1/3}$ Time Algorithm}\label{sec:new}
In this section, we show how to change weights to improve the number of iterations in our algorithm. We will follow the framework of Liu and Sidford \cite{LS20} of finding a set of weights to add under a norm constraint such that the step one would take with respect to the new set of weights maximizes a potential function. In their case, since the step they are taking is an electrical flow, the potential function considered is the energy of such a flow. As our step is different, we will instead take the potential decrement as the potential function with respect to the new set of weights. Perhaps suprisingly however, we can make almost all their arguments go through with minor modifications. Let the initial weights be $w$ and say we would like to add a set of weights $w'$. Then we are interested in maximizing the potential decrement with respect to the new set of weights. This can be seen as similar to designing oracles for multiplicative weight algorithms for two-player games where a player plays a move to penalize the other player the most given their current move. Our algorithm first finds a finds a new set of weights and then takes the potential decrement step with respect to the new weights. Finally, for better control of the congestion vector, we show that one can decrease some of the weight increase like in \cite{LS20}. We first focus on the problem of finding the new set of weights. We are going to introduce a set $r'\in \R^E_{++}$ of "resistances" and will optimize these resistances and then obtain the weights from them.  Let $w$ be the current set of weights and $w'$ be the set of desired changes. Without loss of generality, assume that $\hat{u}_e(f) = \hat{u}_e^+(f)$ and now given a resistance vector $r'$, we define the weight changes as
\begin{align*}
    (w^+_e)'=r'_e(\hat{u}_e^+(f))^2 \text{ and } (w^-_e)'=\frac{(w_e^+)'\hat{u}_e^-(f)}{\hat{u}_e^+(f)}
\end{align*}
This is the same set of weight changes done in \cite{LS20} in the context of energy maximization. This set of weights ensures that our point $(f,y,w)$ is well-coupled with respect to $w+w'$ as well, i.e., \begin{align*}\frac{(w_e^+)'}{\hat{u}^+_e(f)} = \frac{(w_e^-)'}{\hat{u}^-_e(f)}\end{align*}
The problem we would now like to solve is
\begin{align*}
    g(W) = \max\limits_{r'>0, \|r'\|_1\leq W}\min\limits_{B^\top \hat{f}=\delta\chi} \Delta\Phi_{w+w'}(f,\hat{f})
\end{align*}
Here $w'$ is based on $r'$ in the form written above. While this is the optimization problem we would like to solve, we are unable to do so due to the $\ell_1$ norm constraint on the resistances. We will however be able to solve a relaxed $q$-norm version of the problem. \small
\begin{align*}
     &g_q(W) = \max\limits_{r'>0, \|r'\|_q\leq W}\min\limits_{B^\top \hat{f}=\delta\chi} \Delta\Phi_{w+w'}(f,\hat{f})\\
    &= \max\limits_{r'>0, \|r'\|_q\leq W}\min\limits_{B^\top \hat{f}=\delta\chi} \Delta\Phi_w(f,\hat{f}) -\sum\limits_{e \in E} (w_e^+)'\log\left(1-\frac{f'_e}{u^+-f_e}\right) + (w_e^-)' \log\left(1+\frac{f'_e}{u^+-f_e}\right) - f'_e\left(\frac{(w_e^+)'}{u_e^+-f_e} - \frac{(w_e^-)'}{u_e^-+f_e}\right)
\end{align*}\normalsize
Notice that this is a linear (and hence concave) function in $w'$ and hence in $r'$ and is closed and convex in $\hat{f}$ and the constraints are convex as they are only linear and norm ball constraints. Hence, using Theorem \ref{thm:minmaxthm}, we can say that $$\min\limits_{B^\top \hat{f}=\delta\chi} \Delta\Phi_w(f,\hat{f}) -\sum\limits_{e \in E} (w_e^+)'\log\left(1-\frac{f'_e}{u^+-f_e}\right) + (w_e^-)' \log\left(1+\frac{f'_e}{u^+-f_e}\right) - f'_e\left(\frac{(w_e^+)'}{u_e^+-f_e} - \frac{(w_e^-)'}{u_e^-+f_e}\right)$$ is concave in $r'$ and $$\max\limits_{r'>0, \|r'\|_q\leq W} \Delta\Phi_w(f,\hat{f}) -\sum\limits_{e \in E} (w_e^+)'\log\left(1-\frac{f'_e}{u^+-f_e}\right) + (w_e^-)' \log\left(1+\frac{f'_e}{u^+-f_e}\right) - f'_e\left(\frac{(w_e^+)'}{u_e^+-f_e} - \frac{(w_e^-)'}{u_e^-+f_e}\right)$$ is convex in $\hat{f}$. Now, as in \cite{LS20}, we use Sion's minimax lemma to get
\small\begin{align*}
    &\min\limits_{B^\top \hat{f}=\delta\chi}\Delta\Phi_w(f,\hat{f})+\max\limits_{r'>0, \|r'\|_q\leq W}  -\sum\limits_{e \in E} (w_e^+)'\log\left(1-\frac{\hat{f}_e}{u^+-f_e}\right) + (w_e^-)' \log\left(1+\frac{\hat{f}_e}{u^+-f_e}\right) - \hat{f}_e\left(\frac{(w_e^+)'}{u_e^+-f_e} - \frac{(w_e^-)'}{u_e^-+f_e}\right)\end{align*}\normalsize
   \begin{equation}\label{eqn:minmax} \min\limits_{B^\top \hat{f}=\delta\chi}\Delta\Phi_w(f,\hat{f})+W \left[\sum\limits_{e\in E}g_e(\hat{f})^p\right]^{1/p}
   \end{equation}
   where $g_e(\hat{f}) = (\hat{u}_e^+(f))^2\log\left(1-\frac{\hat{f}_e}{\hat{u}^+}\right) + \hat{u}_e^+(f)\hat{u}_e^-(f) \log\left(1+\frac{\hat{f}_e}{\hat{u}_e^-(f)}\right)$ and we plugged in the value of $w'$ in terms of $r'$  and used that $\max\limits_{\|x\|_q\leq W} y^\top x = W\|y\|_p $ with $1/p + 1/q =1$. As mentioned above, the function inside the minimization problem is convex. Furthermore, from the proof of Theorem \ref{thm:minmaxthm}, it can be inferred that any smoothness and strong convexity properties that the function inside the min-max had carries over on the function after the maximization. Hence as in Section \ref{sec:warmup}, we will consider the quadratic extension of the function (as a function of $f$ for the function inside the min-max with $\ell_e = \hat{u}_e(f)/10$. This is just the quadratic extension of $\Delta\Phi_w(f,\hat{f})$ and the quadratic extension of $g_e(f)$.  Now, the strategy will be to consider adding flow using this step while the remaining flow to be routed $F^*-F\geq m^{1/2-\eta}$. After which, running $m^{1/2-\eta}$ iterations of augmenting paths gets us to the optimal solution. We will need to ensure that that throughout the course of the algorithm the $\ell_1$ norm of the weights doesnt get too large. For doing that, we will first compute the weight changes and then do a weight reduction procedure \cite{LS20} in order to always ensure that $\|w\|_1\leq 3m$.

We will take $\eta = 1/6-o(1)-\frac{1}{3}\log_m(U)$ and $W=m^{6\eta}$. Provided we can ensure that the $\|w\|_1\leq 3m$ throughout the course of the algorithm, that the $\ell_\infty$ of the congestion vector is always bounded by a constant and that we can solve the resulting step in almost-linear time, we will obtain an algorithm which runs in time $m^{4/3+o(1)}U^{1/3}$ time.
\begin{theorem} There exists an algorithm for solving $s-t$ maximum flow in directed graphs in time $m^{4/3 + o(1)}U^{1/3}$ time.
\end{theorem}
To summarize, our algorithm starts off with $(f,y)=(0,0)$ and $w_e^+=w_e^-=1$ for all edges $e$. Then in each iteration, starting with a well-coupled $(f,y,w)$ with flow value $F$ and $\delta=(F^*-F)/m^{1/2-\eta}$ and $W=m^{6\eta}$ we then solve Equation \ref{eqn:minmax} (which is the potential decrement problem with the new weights) problem to obtain $\hat{f}$ which will be the step we will take (and has flow value $F+\delta$ and then all that remains is to actually find the update weights $w'$ which will have a closed form expression in terms of $\hat{f}$ and then we perform a weight reduction step to obtain the new $w'$ which ensures that we still remain well-coupled for $\hat{f}$ and repeat while $F^*-F\geq m^{1/2-\eta}$. Finally, we round the remaining flow using $m^{1/2-\eta}$ iterations of augmenting paths. We first state the lemma the proof of which is similar to Lemma \ref{lem:constcongrootm}
\begin{lemma}\label{lem:cong}Let $(f,y,w)$ be a well-coupled solution with value F and let $\delta=\frac{F^*-F}{5000m^{1/2-\eta}}$. Let $\hat{f}$ be the solution to the potential decrement problem considered in Equation \ref{eqn:minmax}. Then, we have for all edges $e$ that $\rho^+_e,\rho^-_e\leq 0.1$ and $|\hat{f}_e|\leq 9m^{-2\eta}$
\end{lemma}
We will prove this lemma in the Appendix \ref{app:missingProofs}. Next notice that $(f,y)$ are still a well-coupled solution with respect to the new weights $w+w'$ as the weights were chosen to ensure that the coupling condition is unchanged.
\begin{lemma}\label{lem:wtcontrol}
Our new weights, after weight reduction, satisfy $\|w'\|_1\leq m^{4\eta+o(1)}U\leq m/2$ and $(f+\hat{f},y+\hat{y})$ is well-coupled with respect to $w+w'$
\end{lemma}
\begin{proof}
Using optimality conditions of the program in Equation \ref{eqn:minmax}, we see that there exists a $\hat{y}$ such that
\begin{align*}
    \hat{y}_v-\hat{y}_u= \hat{f}_e \left(\frac{w_e^+}{(\hat{u}^+_e-\hat{f}_e)\hat{u}_e^+} - \frac{w_e^-}{(\hat{u}_e^-+\hat{f}_e)\hat{u}_e^-}\right)+W\hat{f}_e\frac{g_e^{p-1}}{\|g\|_p^{p-1}}\left(\frac{\hat{u}_e^+}{\hat{u}^+_e-\hat{f}_e}-\frac{\hat{u}^-_e}{\hat{u}^-_e+\hat{f}_e} \right)
\end{align*}
where $g\in\R^E$ is the vector formed by taking $g_e(\hat{f})$ for the $e^{th}$ coordinate. We will take \begin{align*}
   (r_e)'=W\frac{g_e^{p-1}}{\|g\|_p^{p-1}} \text{ and } (w_e^+)' = W\frac{g_e^{p-1}}{\|g\|_p^{p-1}}(\hat{u}_e^+)^2 \text{ and } (w_e^-)' = W\frac{g_e^{p-1}}{\|g\|_p^{p-1}}(\hat{u}_e^+\hat{u}_e^-)
\end{align*}
which satisfies the well-coupling condition we want to ensure. Also notice that $\|r\|_q=W$ so we satisfy the norm ball condition as well. Now, we need to upper bound the $\ell_1$ norm of $w'$. We will take $p=\sqrt{\log m}$
\begin{align*}
    \|w'\|_1&\leq m^{1/p}\|w'\|_q\\
    &\leq m^{o(1)}\left(\sum\limits_{e \in E}(w_e^+)'+(w_e^-)'\right)^{1/q}\\
    &\leq 2m^{o(1)} WU^2=O(m^{6\eta+o(1)}U^2)
\end{align*}
as $\hat{u}_e^+,\hat{u}_e^-\leq U$. Plugging in the value of $\eta$, we get that this is less than $m/2$. Now, we will perform weight reductions to obtain a new set of weights $w''$ such that they still ensure the coupling condition doesnt change and we can establish better control on the weights. The weight reduction is procedure is the same as that in \cite{LS20} where we find the smallest non-negative $w''$ such that for all edges
\begin{align*}
    \frac{(w_e^+)'}{\hat{u}_e^+-\hat{f}_e}-\frac{(w_e^-)'}{\hat{u}_e^-+\hat{f}_e}=\frac{(w_e^+)''}{\hat{u}_e^+-\hat{f}_e}-\frac{(w_e^-)''}{\hat{u}_e^-+\hat{f}_e}
\end{align*}
Notice that we also have that $
    \frac{(w_e^+)'}{\hat{u}_e^+}=\frac{(w_e^-)'}{\hat{u}_e^-}
$ and
\begin{align*}
    \frac{\hat{u}_e^+-\hat{f}_e}{\hat{u}_e^-+\hat{f}_e}=(1\pm O(\max\{\rho_e^+,\rho_e^-\})\frac{\hat{u}_e^+}{\hat{u}_e^-}
\end{align*}
Hence, it follows that
\begin{align*}
    (w_e^+)'' + (w_e^-)'' \leq O(\max\{\rho_e^+,\rho_e^-\}) ((w_e^+)' + (w_e^-)')
\end{align*}
As $|\hat{f}_e|\leq 9m^{-2\eta}$ from Lemma \ref{lem:cong}, we get
\begin{align*}
    \|w''\|_1&\leq m^{-2\eta}\sum\limits_{e \in E}\frac{W g_e^{p-1}}{\|g\|_p^{p-1}}(\hat{u}_e^++\hat{u}_e^-)\\
    &\leq O(m^{4\eta+o(1)}U)\leq m/2
\end{align*}

As before, while this argument is done for the non-quadratically extended function while we are optimizing the quadartically extended function, as our $\rho^+_e,\rho^-_e\leq 0.1$, the minimizers are the same and hence the above argument works.
\end{proof}
Now, provided that we can show how to solve Equation \ref{eqn:minmax} in almost-linear time, we are done. This is because we run the algorithm for $m^{1/2-\eta}$ iterations and the $\ell_1$ norm of the weights increases by at most $m^{4\eta+o(1)}U$ in each iteration. Hence the final weights are $\|w\|_1 \leq 2m + m^{1/2+3\eta+o(1)}U\leq 5m/2$. So we can use Lemma \ref{lem:precond} throughout the course of our algorithm. Also, as mentioned above, notice that the flow $\hat{f}$ that we augment in every iteration is just the solution to the potential decrement problem with the new weights. Hence, from the argument in Section \ref{sec:warmup}, we always maintain the well-coupled condition.

To show that we can solve the problem in Equation \ref{eqn:minmax}, we will appeal to the work of \cite{KPSW19}. As mentioned above, their work establishes Lemma \ref{lem:kpsw} and then shows that for any function  which can be sandwiched in that form plus a quadratic term which is the same on both sides, one can just minimize the resulting upper bound to get a solution to the optimization problem with quasi-polynomially low error. Hence, we will focus on showing that the objective function in our problem can also be sandwiched into terms of this form after which appealing to their algorithm, we will get a high accuracy solution to our problem in almost linear time. The first issue that arises is that srictly speaking, their algorithm only works for minimizing objectives of the form
\begin{align*}
    OPT=\min\limits_{B^\top f=\chi} \sum\limits_{e \in E} g_e f_e + r_e f_e^2 + |f_e|^p
\end{align*}
whereas for our objective, the $p$-norm part is not raised to the power $p$ but is just the $p$-norm itself. The solution for this however was already given in Liu-Sidford \cite{LS20} where they show (Lemma B.3 in their paper) that for sufficiently nice functions minimizing problems of the form $\min\limits f(x)+h(g(x))$ can be obtained to high accuracy if we can obtain minimizers to functions of the form $f(x)+g(x)$. The conditions they require on the functions are also satisfied for our functions and is a straightforward calculation following the proof in their paper \cite{LS20}. Hence, we can focus on just showing how to solve the following problem
\begin{align*}
    OPT=\min\limits_{B^\top \hat{f} = \chi} \sum\limits_{e\in E} -\left(w_e^+\log_{0.1}\left(1-\frac{\hat{f}_e}{\hat{u}_e^+}\right) + w_e^-\log_{0.1}\left(1+\frac{\hat{f}_e}{\hat{u}_e^-}\right) + \hat{f}_e\left(\frac{w_e^+}{\hat{u}_e^+} - \frac{w_e^+}{\hat{u}_e^-} \right)\right)+ (g_e)_{0.1}(\hat{f})^p
\end{align*}
Where the subscripts of $0.1$ denote that we are solving the quadratically smoothened function with the box size being $\hat{u}_e(f)/10$ for each $e$ and $g_e(\hat{f}) = (\hat{u}_e^+(f))^2\log\left(1-\frac{\hat{f}_e}{\hat{u}^+}\right) + \hat{u}_e^+(f)\hat{u}_e^-(f) \log\left(1+\frac{\hat{f}_e}{\hat{u}_e^-(f)}\right)$ Call the term in the sum for a given edge $e$ as $val_e(\hat{f})$ and the overall objective function is $val(\hat{f})$. In particular,  we consider for a single edge and prove the following lemma
\begin{lemma}\label{lem:sandwhich}
We have the following for any feasible $f$ and $\delta \geq 0$
\begin{align*}
    val_e(f) + \delta\partial_f val_e(f)+ (9/10)^2\delta^2\left(\frac{w_e^+}{(\hat{u}_e^+-f)^2}+\frac{w_e^-}{(\hat{u}_e^-+f)^2}\right) + 2^{-O(p)}(f_e^{2p-4}\delta^2 + \delta^p) \leq val_e(f+\delta)
\end{align*}
and
\begin{align*}
val_e(f+\delta) \leq val_e(f) + \delta\partial_f val_e(f)+(11/10)^2 \delta^2\left(\frac{w_e^+}{(\hat{u}_e^+-f)^2}+\frac{w_e^-}{(\hat{u}_e^-+f)^2}\right) + 2^{O(p)}(f_e^{2p-4}\delta^2 + \delta^p)
\end{align*}
where $\partial_x$ denotes the derivative of a function with respect to $x$.
\end{lemma}
We prove this lemma in Appendix \ref{app:missingProofs}. Let $r_e = \left(\frac{w_e^+}{(\hat{u}_e^+-f)^2}+\frac{w_e^-}{(\hat{u}_e^-+f)^2}\right)$
\begin{lemma}\label{lem:relconv} Now, given an initial point $f_0$ such that $B^\top f_0 = \chi$ and an almost linear time solver for the following problem
\begin{align*}
    \min\limits_{B^\top \delta = 0} \sum\limits_{e \in E}
    \delta_e \alpha_e + (11/10)^22^{O(p)}((r_e+f_e^{2p-4})\delta^2 + \delta^p)
\end{align*}
where the $\alpha_e$ vector is the gradient of $val$ at a given point $f$, we can obtain an $\hat{f}$ in $\widetilde{O}_p(1)$ calls to the solver such that $val(\hat{f})\leq OPT+ 1/2^{\poly\log m}$
\end{lemma}
The proof is similar to the proof of the iteration complexity of gradient descent for smooth and strongly convex function and it follows from \cite{LFN18,KPSW19}.
Note that since \cite{KPSW19} give an almost linear time solver for exactly the subproblem in the above lemma provided the resistances are quasipolynomially bounded, we are done. This is because Section D.1 in \cite{LS20} already proves that the resistances are quasipolynomially bounded.
\section{Conclusion}\label{sec:conclusion}
In this paper, we showed how to use steps inspired by potential reduction IPMs to solve max flow in directed graphs in  $O(m^{4/3+o(1)}U^{1/3})$ time. We believe our framework for taking the step corresponding to the maximum decrease of the potential function may be useful for other problems including $\ell_p$ norm minimization. In particular, can one set up a homotopy path for which steps are taken according to a potential function. Presumably if this can be done, this might also offer hints for how to use ideas corresponding to different homotopy paths induced by other potential functions (rather than the central path we consider) to solve max flow faster. Finally, there is no reason to believe that the procedure for selecting weight changes corresponding to the potential decrement being maximized to be the best way to change weights. This may lead to a faster algorithm as well if one can find another strategy which establishes tighter control on weight changes. A question along the way to such a strategy might be to understand how the potential decrement optimum changes as we change weights/resistances. Such an analog for change in energy of electrical flow as we change resistances is used in \cite{CKMST11,M16,LS20}. Another open problem that remains is obtaining faster algorithms for max flow on weighted graphs with logarithmic dependence on $U$ as opposed to the polynomial dependence in this paper.

\section*{Acknowledgements}
We would like to thank Jelena Diakonikolas, Yin Tat Lee, Yang Liu, Aaron Sidford and Daniel Spielman for helpful discussions. We also thank Jelani Nelson for several helpful suggestions regarding the presentation of the paper.


  \addcontentsline{toc}{section}{References}
  \bibliographystyle{amsalpha}
  \bibliography{main}

\appendix


\section{Missing Proofs}\label{app:missingProofs}
\begin{proof}{[of Lemma \ref{lem:cong}]} We follow the strategy used in the proof of Lemma \ref{lem:constcongrootm}. Recall that the problem we are trying to understand is
\begin{align*}
    \min\limits_{B^\top \hat{f}=\delta\chi}\Delta\Phi_w(f,\hat{f}) + W\left[\sum\limits_{e\in E}g_e(\hat{f})^p\right]^{1/p}
\end{align*}
where $g_e(\hat{f}) = (\hat{u}_e^+(f))^2\log\left(1-\frac{\hat{f}_e(f)}{\hat{u}^+_e(f)}\right) + \hat{u}_e^+(f)\hat{u}_e^-(f) \log\left(1+\frac{\hat{f}_e}{\hat{u}_e^-(f)}\right)$. As in Lemma \ref{lem:constcongrootm}, we will consider a flow $f'$ which sends $\frac{2\delta}{m}$ units of flow on each of the $m/2$ preconditioned edges. Certainly, the objective value of the above function at $\hat{f}$ will have a smaller value than that at $f'$. For the first term $\Delta\Phi_w(f,\hat{f})$, running the same argument as in Lemma \ref{lem:constcongrootm}, we get that \begin{align*}\Delta\Phi_w(f,\hat{f}) &\leq \|w\|_1 \left(\frac{42\delta}{F^*-F}\right)^2\\
&\leq 0.000071m^{2\eta}
\end{align*}
For the the second term, we use $\log(1-x)\leq -x + x^2$ and $\log(1+x) \leq x+x^2$, to get that $g_e(f)\leq \hat{f}_e^2\left(1+\frac{\hat{u}_e^+(f)}{\hat{u}_e^+(f)}\right)\leq 2\hat{f}_e^2$ where we have used that $\hat{u}^+_e(f)\leq \hat{u}_e^-(f)$. Now, since there is non-zero flow on the preconditioned edges, we get that
\begin{align*}
    W\left[\sum\limits_{e\in E}g_e(f')^p\right]^{1/p}&\leq 2W (\delta/m)^2 m^{o(1)}\\
    &\leq 2m^{6\eta+o(1)}\left(\frac{F^*-F}{5000m(m^{1/2-\eta})}\right)^2\\
    &\leq 0.0000004m^{8\eta-1+o(1)}U^2
\end{align*}
using $p=\sqrt{\log n}$, the fact that $F^*-F\leq mU$ and the value of $\delta=\frac{F^*-F}{5000m^{1/2-\eta}}$. Also using the value of $\eta$, we can see that this term is less than $0.0000001m^{2\eta}$.
Hence, combining the two, we get that the objective value at $\hat{f}$ is less than $0.000072m^{2\eta}$. As the objective function is made up of two non-negative quantities, we can obtain two inequalities using this upper bound by dropping one term from the objective value each time. For the second part, we ignore the first term of the objective function and lower bound the second term using the fact that $\log(1+x) \geq x+x^2/2$ and $\log(1-x)\geq -x+x^2/2$
\begin{align*}
0.000072m^{2\eta}&\geq W\left[\sum\limits_{e\in E}g_e(\hat{f})^p\right]^{1/p}\geq W|g_e(\hat{f})|\\
&\geq W \hat{f}_e^2 (1+\hat{u}_e^+(f)/\hat{u}_e(f))\\
&\geq W\hat{f}_e^2
\end{align*}
This gives us that $|\hat{f}_e|\leq 0.009 m^{-2\eta}$ by plugging in the value of $W=m^{6\eta}$

For the first part now, assume for the sake of contradiction that $\rho_e> 0.1$, otherwise we are done. Now, dropping the second term we want to establish that $\frac{1}{\hat{u}_e(f)} \leq 9 m^{2\eta}$, which we will do so by a proof similar to the proof of Lemma 4.3 in \cite{M16}. Now using the argument as in Lemma \ref{lem:constcongrootm}, we get for an edge $e=(u,v)$,

\begin{align*}
  0.000072m^{2\eta}&\geq \Delta\Phi_w(f,\hat{f})\\
  &\geq\frac{1}{2.2}\sum\limits_{e \in E} \hat{f}_e \left(\frac{w_e^+}{u_e^+-f_e-\hat{f}_e} - \frac{w_e^-}{u_e^-+f_e+\hat{f}_e} - \frac{w_e^+}{u_e^+-f_e} + \frac{w_e^-}{u_e^-+f_e}\right)\\
  &=\frac{1}{2.2}\hat{f}^\top B\hat{y}\\
  &=\frac{1}{2.2}\delta \chi^\top\hat{y}\\
  &=\frac{F^*-F}{11000m^{1/2-\eta}} \chi^\top \hat{y}\\
  &\geq \chi^\top \hat{y}/11000\\
  &=(\hat{y}_s-\hat{y}_t)/11000\\
  &\geq (\hat{y}_u-\hat{y}_v)/11000\\
  &= \frac{1}{11000}\left(\frac{w_e^+}{u_e^+-f_e-\hat{f}_e} - \frac{w_e^-}{u_e^-+f_e+\hat{f}_e} - \frac{w_e^+}{u_e^+-f_e} + \frac{w_e^-}{u_e^-+f_e}\right)\\
  &\geq \frac{9\hat{f}_e}{110000}\left(\frac{w_e^+}{(u_e^+-f_e)^2} + \frac{w_e^-}{(u_e^-+f_e)^2}  \right)\\
  &\geq \frac{9\rho_e}{110000\hat{u}_e(f)}\\
  &\geq \frac{0.9}{110000\hat{u}_e(f)}
\end{align*}
where the first and second equalities follows from optimality and feasibility conditions of the potential decrement problem respectively and the third inequality follows from the condition that we run the program while the flow left to augment is at least $m^{1/2-\eta}$. This implies that $1/\hat{u}_e(f)\leq 9m^{2\eta}$. Multiplying this with $|\hat{f}_e|\leq 0.009m^{-2\eta}$, we get that $\rho_e \leq 0.1$, which finishes the proof. We also need to argue the inequality $\hat{y}_s-\hat{y}_t\geq \hat{y}_u-\hat{y}_v$.  The optimality conditions of $$\hat{y}_u-\hat{y}_v=\hat{f}_e \left(\frac{w_e^+}{(u_e^+-f_e-\hat{f}_e)(u_e-f_e)} + \frac{w_e^-}{(u_e^-+f_e)(u_e^-+f_e+\hat{f}_e)}\right)$$
and noticing that the quantity in brackets in the right hand side above is non-negative, tells us that there is a fall in potential along the flow. This along with noticing that the sum of the potential difference in a directed cycle is zero, tells us that the graph induced by just the flow $\hat{f}$ is a DAG. Since, it's a DAG, it can be decomposed into disjoint $s-t$ paths along which flow is sent and every edge belongs to one of these paths. Hence, the potential difference across an edge is less than the potential difference across the whole path which is the potential difference between $s$ and $t$ and hence, we are done.

As before, all these arguments go through with the quadratically smoothened cases cases rather than the original function to still get the same bounds and since $\rho_e \leq 0.1$, the minimizers of the two are the same which completes the proof.
\end{proof}
\begin{proof}{[of Lemma \ref{lem:sandwhich}]}
Note that while we are solving for the quadratically smoothened version of the problem, we can assume we solve it for the non-smoothened version in the box corresponding to a congestion of at most $0.1$ as the extension is $C^2$ and  will ensure that any inequalities we need henceforth (upto the second order terms) are bounded as well.

There are two terms, one corresponding to the potential decrement term and the other is a similar expression but raised to the $p^{th}$ power. We tackle the first term first. This is easily done using Taylor's theorem. The function is $g(x+y) = -\log(1-(x+y)/u) - (x+y)/u$. Computing the first two derivatives with respect to $y$, we get that
$g'(x+y) = \frac{1}{u-x-y} - 1/u$ and $g''(x+y) = \frac{1}{(u-x-y)^2}$. Now, using Taylor's theorem, we get that
\begin{align*}
    g(x+y) &= g(x) + g'(x)y+\frac{1}{2}g''(x+\zeta)y^2\\
    &= g(x) + y\left(\frac{1}{u-x-y}-\frac{1}{u}\right) + y^2 \left(\frac{1}{(u-x-\zeta)^2}\right)
\end{align*}
for some $ \zeta$ such that $-u/10\leq x+\zeta \leq u/10$ which easily gives us the bound
\begin{align*}
    g(x) + y\left(\frac{1}{u-x-y}-\frac{1}{u}\right) + (9/10)^2 y^2 \left(\frac{1}{(u-x)^2}\right) \leq g(x+y) \leq g(x) + y\left(\frac{1}{u-x-y}-\frac{1}{u}\right) + (11/10)^2 y^2 \left(\frac{1}{(u-x)^2}\right)
\end{align*}
Similarly for $-\log(1+x/u)+x/u$.

Now, for the second term, we will largely follow the strategy of \cite{KPSW19}.
Now for the $p^{th}$ order term, we have a function
$g(x)=u_1^2 \log(1-x/u_1) + u_1u_2 \log(1-x/u_2)$. We first use Lemma \ref{lem:kpsw} with $f_i = g(x)$ and $\delta_i = g(x+y)-g(x)$ to get
\begin{align*}
    g(x+y)^p &\leq g(x)^p + p g(x)^{p-1}(g(x+y)-g(x)) + 2^{O(p)}(g(x)^{p-2}(g(x+y)-g(x))^2 + (g(x+y)-g(x))^p)
\end{align*}
Now, adding and subtracting $pg(x)^{p-1}yg'(x)$ from both sides and noticing that $g(x+y)-g(x)-yg'(x) \leq 0$ from concavity of $g$ , we get
\begin{align*}
g(x+y)^p &\leq g(x)^p + pyg(x)^{p-1}g'(x) + 2^{O(p)}(g(x)^{p-2}(g(x+y)-g(x))^2 + (g(x+y)-g(x))^p)
\end{align*}
Now, notice that using inequalities of $\log(1-x/u)$ and $\log(1+x/u)$, to get $x^2 \leq g(x) \leq 2 x^2$ and we also use Taylor's theorem get that $g(x+y)-g(x) \leq 10(|xy| + |y|^2)$
\begin{align*}
    g(x+y)^p &\leq g(x)^p + p yg(x)^{p-1}g'(x) + 2^{O(p)}(x^{2p-4}(x^2 y^2 + y^4) + 2^{p-1}(x^py^p + y^{2p})\\
    &\leq g(x)^p + p yg'(x) + 2^{O(p)}(x^{2p-2}y^2 + y^{2p})
\end{align*}
where we have used $(x+y)^p \leq 2^{p-1}(x^p + y^p)$ and that $y \leq x$ because that's the neighborhood we are considering. Beyond that neighborhood, we could just do the calculation with the quadratic extension parts (as we only used upto the second order information so that's still preserved)

The proof of the lower bound is similar.
\end{proof}

\end{document}